\documentclass[conference]{IEEEtran}
\IEEEoverridecommandlockouts
\usepackage{cite}
\usepackage{amsmath,amssymb,amsfonts,latexsym}
\usepackage{amsthm}
\usepackage{braket}
\usepackage{graphicx}
\usepackage{textcomp}
\usepackage{xcolor}
\usepackage{mathtools}
\usepackage{multirow}
\usepackage{url}
\usepackage{booktabs}
\usepackage{hyperref}
\usepackage{lipsum}
\usepackage{nicematrix}
\usepackage[all]{xy}
\usepackage{tikz-cd}
\usepackage{tikz}
\usepackage{makecell}
\usepackage{caption}
\usepackage{subcaption}
\usepackage{quantikz}

\usetikzlibrary{calc}
\usetikzlibrary{intersections,shapes.arrows}
\usetikzlibrary{arrows.meta, positioning, backgrounds}

\usepgfmodule{nonlineartransformations}
\def\fluttertransform{%
    \pgfgetlastxy\x\y
    \pgfpoint{\x+sin(\y)}{\y+sin(\x)*(30-\x/2)+\x/10}
}

\def\wtilde{\widetilde}

\def \mcal{\mathcal}

\newtheorem*{theorem*}{Theorem}
\newtheorem*{lemma*}{Lemma}
\newtheorem*{example*}{Example}

\makeatletter
\renewenvironment{proof}[1][\proofname]{%
  \par\pushQED{\qed}\normalfont
  \topsep6\p@\@plus6\p@\relax%
  \trivlist\item[\hskip\labelsep\bfseries #1\@addpunct{.}]\ignorespaces
}{%
  \popQED\endtrivlist\@endpefalse
}
\makeatother

\def\BibTeX{{\rm B\kern-.05em{\sc i\kern-.025em b}\kern-.08em
    T\kern-.1667em\lower.7ex\hbox{E}\kern-.125emX}}

\makeatletter
\newcommand{\linebreakand}{%
  \end{@IEEEauthorhalign}
  \hfill\mbox{}\par
  \mbox{}\hfill\begin{@IEEEauthorhalign}
}
\makeatother

\begin{document}

\title{Diagonal Adaptive Non-local Observables on Quantum Neural Networks\\
\thanks{
This work is supported by the U.S. DOE, under award DE-SC-0012704 and BNL's LDRD \#24-061. \textbf{Code available: \url{https://github.com/HHTseng/DiagonalANO}}}
}

\author{
\IEEEauthorblockN{Huan-Hsin~Tseng}
\IEEEauthorblockA{\textit{AI \& ML Department} \\
\textit{Brookhaven National Laboratory}\\
Upton NY, USA  \\
htseng@bnl.gov}
\and
\IEEEauthorblockN{Yan Li}
\IEEEauthorblockA{\textit{Department of Electrical Engineering} \\
\textit{The Pennsylvania State University}\\
University Park, PA, USA \\
yql5925@psu.edu}
\and
\IEEEauthorblockN{Hsin-Yi~Lin}
\IEEEauthorblockA{\textit{AI \& ML Department} \\
\textit{Brookhaven National Laboratory}\\
Upton NY, USA  \\
hlin7@bnl.gov}
\linebreakand
\IEEEauthorblockN{Samuel Yen-Chi~Chen}
\IEEEauthorblockA{\textit{Wells Fargo} \\
New York NY, USA \\
ycchen1989@ieee.org}
}

\maketitle

\begin{abstract}
Adaptive Non-local Observables (ANOs) have shown that making quantum observables dynamic can substantially enlarge the function space of Variational Quantum Algorithms, partly shifting hardware demands from circuit synthesis to measurement design. However, this advantage is accompanied by a steep increase in the number of parameters, as well as the classical optimization cost for varying general Hermitian observables.

We propose a special form of ANO that significantly reduces this burden by considering only diagonal observables paired with quantum circuits. Mathematically, this is equivalent to the full ANO of a large parameter space since diagonal matrices are canonical representatives of the ANO space modulo unitary similarity. As a result, Diagonal ANO retains the same capability of full ANO while reducing $k$-local observable complexity from $O(4^k)$ to $O(2^k)$ and lowering the corresponding measurement-side classical computation. In this sense, diagonal ANO preserves much of the benefit of full ANO while encompassing conventional VQCs as a special case.
\end{abstract}

\begin{IEEEkeywords}
Variational Quantum Algorithms, Quantum Machine Learning, Non-local observables, Heisenberg picture.
\end{IEEEkeywords}

\section{Introduction} \label{sec:Indroduction}
Recent progress toward fault-tolerant hardware, including below-threshold surface-code memories and programmable logical processing, reinforces the long-term promise of quantum computing while underscoring that near-term devices remain resource-constrained. Hybrid Quantum-Classical approaches are therefore actively pursued as a practical route to extract utility from noisy intermediate-scale quantum (NISQ) processors \cite{preskill2018quantum, cerezo2021variational, bharti2022noisy}.

A central hybrid paradigm is the variational quantum algorithm (VQA), where a parameterized quantum circuit is trained in a classical optimization loop \cite{mcclean2016theory, cerezo2021variational}. VQAs support applications such as the variational quantum eigensolver for many-body and molecular problems \cite{peruzzo2014variational, kandala2017hardware}, and they also provide the backbone for quantum machine-learning models (often termed variational quantum circuits or quantum neural networks) in supervised learning and related tasks \cite{biamonte2017quantum}.

In these models, a prediction typically takes the form of an expectation value $f_{\theta}(x)=\langle 0|U^{\dagger}(x,\theta)\, H \,U(x,\theta)|0\rangle$, where the readout observable $H$ is usually fixed (commonly a local Pauli operator or a simple sum thereof). Such hypothesis class $f_{\theta}$ depends jointly on the data encoding, the trainable circuit ansatz $U(x,\theta)$ of parameters $\theta$, and the family of allowed observables. 

Such encoded variational quantum models result in a partial Fourier series $f_{\theta}(x)=\sum_{\omega\in\Omega} c_{\omega}(\theta)e^{i\omega x}$, where the encoding determines the accessible frequency set $\Omega$, while the trainable circuit and the \textit{choice of observable} determine the Fourier coefficients $c_{\omega}(\theta)$ to approximate the target function of the supervised task~\cite{schuld2021effect}. This motivates enlarging the model class not only through the circuit ansatz, but also through the choice of measurement observable.

Quantum machine learning (QML) investigates learning models that incorporate quantum resources, with quantum neural networks (QNNs) commonly realized as variational quantum circuits (VQCs), i.e., parameterized quantum circuits trained via classical optimization \cite{biamonte2017quantum,cerezo2021variational}. These models have been applied to tasks including classification, convolutional architectures, and Reinforcement Learning \cite{perez2020data,chen2022quantum,chen2023asynchronous}. Their performance depends on ansatz expressivity and trainability \cite{sim2019expressibility,mcclean2018barren}, while recent work treats measurement as a trainable component that can further enhance model capacity \cite{chen2021end,chen2025learning}. This measurement-centric perspective directly motivates Adaptive Non-local Observables (ANOs)~\cite{lin2025adaptive}.

ANO addresses this issue explicitly by considering a trainable multi-qubit Hermitian observable, inspired by the Heisenberg picture viewpoint that complexity can be shifted from state preparation to observables. Beyond the original classification task setting, this adaptive-measurement paradigm has also been explored in broader learning tasks such as Reinforcement Learning and Super-Resolution~\cite{lin2025quantumANORL, lin2026anosr}. A drawback is scaling; a general Hermitian observable on $k$ qubits carries $\mathcal{O}(4^k)$ degrees of freedom, which may quickly increase the parameter amount.

This work proposes \emph{Diagonal Adaptive Non-local Observables} (DANO), which couples a variational circuit with a diagonal trainable observable. The key observation is that any Hermitian operator is unitarily diagonalizable. In the ANO setting, this diagonal restriction can be compensated by surrounding circuit unitaries when the available gate set and circuit depth are sufficiently expressive to approximate the relevant diagonalizing transformations. Consequently, DANO reduces Hermitian parameters from $\mathcal{O}(4^k)$ to $\mathcal{O}(2^k)$ while retaining the ability to approach full ANO expressivity.

\section{Background}\label{Sec: Background}

\subsection{Variational Quantum Circuits (VQC)}\label{Subsec: VQC}
The VQC is a member of VQA family to learn classical data of the form $\mathcal{D} = \{(x^{(j)}, y^{(j)}) \, | \, x^{(j)} \in \mathbb{R}^n, y^{(j)} \in \mathbb{R}, \, j \in \mathbb{N} \}$ with $x^{(j)}$ as an input of sample $j$ and $y^{(j)}$ as the ground truth. 

Let $\mcal{H}^n$ be the Hilbert space of $n$-qubit states and $\mcal{U}(\mathcal{H}^n) = \{ U: \mathcal{H}^n \to \mathcal{H}^n \, | \, U U^{\dagger} = U^{\dagger} U = I_{2^n} \}$ be the set of all gates. The VQC is a choice of $(V, U(\theta), H)$ such that (Fig.~\ref{fig: VQC})
    \begin{equation}\label{E: VQC}
        \langle H \rangle_{\text{VQC}} (x; \theta) := \bra{\psi_0} V^{\dagger}(x) U^{\dagger}(\theta) \, H \, U(\theta) V(x) \ket{\psi_0}
    \end{equation}
approximates the ground truth $y \in \mathbb{R}$ of a classical input $x \in \mathbb{R}$, where $V: \mathbb{R}^n \to \mcal{U}(\mcal{H}^n)$ is called an \textit{encoding map}, $U(\theta) \in \mcal{U}(\mcal{H}^n)$ is called the \textit{variational circuit} of tunable parameters $\theta$, and $H: \mcal{H}^n \to \mcal{H}^n$ is a Hermitian operator called \textbf{observable}. 

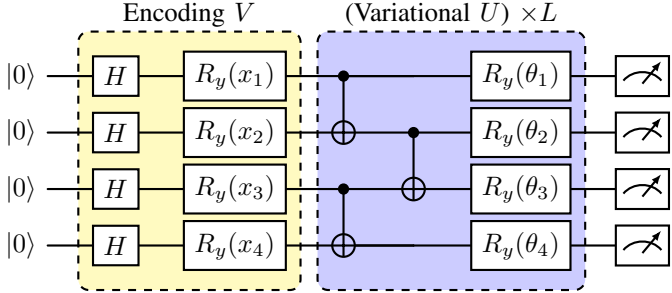
\begin{figure}[htbp]
 \vskip -0.0in
  \centering
    \begin{quantikz}[row sep=0.1cm, column sep=0.6cm]
        \lstick{$\ket{0}$} &\gate{H} \gategroup[wires=4,steps=2,style={dashed,rounded corners,fill=yellow!30,inner xsep=2pt},background]{Encoding $V$} & \gate{R_y(x_1)} & \ctrl{1} \gategroup[wires=4,steps=3,style={dashed,rounded corners,fill=blue!20,inner xsep=2pt},background]{(Variational $U$) $\times L$} & \qw &\gate{R_y(\theta_1)}  & \meter{} \\
        \lstick{$\ket{0}$} &\gate{H} & \gate{R_y(x_2)} & \targ{}  & \ctrl{1} &\gate{R_y(\theta_2)}  & \meter{} \\
        \lstick{$\ket{0}$} &\gate{H} & \gate{R_y(x_3)} & \ctrl{1} & \targ{}  &\gate{R_y(\theta_3)}  & \meter{} \\
        \lstick{$\ket{0}$} &\gate{H}  & \gate{R_y(x_4)} & \targ{}  & \qw      &\gate{R_y(\theta_4)}  & \meter{}
    \end{quantikz}
  \caption{A VQC diagram of (\ref{E: encoding V}), (\ref{E: variational U}), which is also implemented in Sec.~\ref{sec_exp_results}. The variational block represented by a blue box is repeated $L$ times to increase the circuit depth.}
  \label{fig: VQC}
\end{figure}

Typical VQCs \textit{fix} the choice of encoding $V$ and observable denoted by $H_0$ to vary $\theta$ minimizing a loss function, for example,
    \begin{equation}\label{E: loss}
    L(\theta; \mathcal{D}) = \frac{1}{| \mcal{D} |} \sum_{j=1}^{| \mcal{D} |} \| \braket{H_0} (x^{(j)}; \theta) - y^{(j)} \| ^2
    \end{equation}
The training process of VQC optimizing $\theta$ in (\ref{E: loss}) then traces a trajectory $t \mapsto U(\theta^{(t)})$ on the Lie group $\mcal{U}(\mcal{H}^n)$ (Fig.~\ref{fig: Unitary param transition}). It is shown in \cite{lin2025adaptive} that a training curve in $\mcal{U}(\mcal{H}^n)$ can be translated into a curve on the Hermitian operator space $\mathbb{H}(n) = \{H: \mcal{H}^n \to \mcal{H}^n \, | \, H = H^{\dagger}\}$ by mapping $U^{\dagger}(\theta) \mapsto U^{\dagger}(\theta) \, H \, U(\theta)$, which motivated the \textit{Adaptive Non-local Observable} (ANO) approach~\cite{lin2025adaptive}.

\begin{figure}[htbp]
    \centering 
\begin{tikzpicture}[scale=0.9]
\begin{scope}[yshift=20mm]
\pgftransformnonlinear{\fluttertransform}
\draw [fill=red!15] plot [smooth cycle]
coordinates {(-1.14,-1)(-0.84, -.18) (-0.04, 0.3) (2.24, 0) %
(4.48, -0.56) (4.48, -1.46) (3.38,-1.84)(0.38, -1.28)};
\end{scope}

\coordinate (p1) at (0.1, 1.6);
\coordinate (p2) at (1.4, 1.94);
\coordinate (p3) at (2.5, 1.34);
\coordinate (p4) at (3.8, 1.1);

\filldraw (4., 2.5) node[left] {$\mathcal{U}(\mathcal{H}^n)$};

\filldraw (p1) circle (1pt) node[above] {\footnotesize{$U(\vartheta)$}};
\filldraw (p2) circle (1pt) node[above] {\footnotesize{$U(\theta^{(1)})$}};
\filldraw (p3) circle (1pt) node[below] {\footnotesize{$U(\theta^{(2)})$}};
\filldraw (p4) circle (1pt) node[below] {\footnotesize{$U(\widetilde{\vartheta})$}};

\draw[-, dashed] (p1) -- (p2) -- (p3) -- (p4);

\end{tikzpicture}
\caption{The training process searching for optimal circuit $U(\widetilde{\theta})$ results in observable changes $\theta \mapsto H(\theta):=U^{\dagger}(\theta) \, H \, U(\theta)$. }
\label{fig: Unitary param transition}
 \vskip -0.1in
\end{figure}
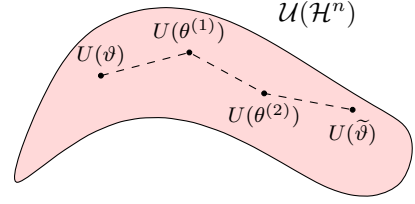

\subsection{Adaptive Non-local Observables (ANO)}\label{Subsec: ANO}

The ANO proposes using a dynamical observable $H$ to enlarge the function space while alleviating circuit depth by
    \begin{equation}\label{E: ANO}
        \langle H \rangle_{\text{ANO}} (x) := \bra{\psi_0} V^{\dagger}(x) \, H \, V(x) \ket{\psi_0}
    \end{equation}
where $H \in \mathbb{H}(n)$ admits an arbitrary $k$-local Hermitian $\wtilde{H} (\phi) \in \mathbb{H}(k)$ with $ k \leq n$ and
\begin{equation}\label{E: non-local Hermitian}
  \wtilde{H}(\phi) =   \begin{pmatrix}
c_{11} & a_{12} + i b_{12} & a_{13} + i b_{13} & \cdots & a_{1 2^k} + i b_{1 2^k}  \\
* & c_{22}  & a_{23} + i b_{23}  & \cdots & a_{2 2^k} + i b_{2 2^k}  \\
* & * & c_{33}  & \cdots & a_{3 2^k} + i b_{3 2^k}  \\
\vdots & \vdots & \vdots & \ddots & \vdots \\
* & * & * & \cdots & c_{2^k 2^k}
\end{pmatrix}
\end{equation}
such that 
\begin{equation}\label{E: H tensor product}
    H = I \otimes \cdots \otimes I \otimes \wtilde{H}(\phi) \otimes I \otimes \cdots \otimes I.
\end{equation}
Denoting $K = 2^k$, a non-local observable can be parametrized by $K^2$ (real) parameters  $\phi = \left( a_{ij}, b_{ij}, c_{ii} \right)_{i, j=1}^K$ of the upper triangle, where the lower triangle in (\ref{E: non-local Hermitian}) is the complex conjugate of the upper such that $\wtilde{H}(\phi) = \wtilde{H}^{\dagger}(\phi)$.

It is shown mathematically ANO alone suffices to contain VQC with variational circuits $U(\theta)$ in (\ref{E: VQC}), \textit{i.e.,} VQC $\subsetneq$ ANO. This is due to the \textit{unitary similarity} (equivalence) relation defined on $\mathbb{H}(n)$ by $H_1 \sim H_2 \iff H_1 = U^{\dagger} \, H_2 \, U$ for some $U \in \mcal{U}(\mcal{H}^n)$ such that $\mathbb{H}(n)$ is divided into equivalent classes $\mathbb{H}(n) / \sim$ (Theorem~\cite{lin2025adaptive}), where VQC only occupies one of them (Fig.~\ref{fig: Hermitian equivalent classes}).

However, as ANO increases the functional variety of circuits, the number of parameters significantly increases. Consequently, we propose a new parameter-efficient variant of ANO inheriting the same flexibility while significantly reducing the number of observable variables.

\section{Diagonal Adaptive Non-local Observables}\label{Sec: DANO}

The proof of Theorem~\cite{lin2025adaptive} reveals that the equivalence classes $\mathbb{H}(n)/\sim$ can be identified with $\mathbb{R}^{2^n}$ via the canonical representatives in diagonal form (Fig.~\ref{fig: Hermitian equivalent classes}).
\begin{figure}[htbp]
    \centering 
\begin{tikzpicture}[scale=1]
  \def\Mcurve{
    plot [smooth cycle] coordinates {(0,0) (3,0.3) (3.5,1.5) (3,2.7) (0,3) (-1,1.5)}
  }

  \begin{scope}
    \clip \Mcurve;
    \fill[blue!20, opacity=0.6] (-2,0) rectangle (5,1);
  \end{scope}

  \begin{scope}
    \clip \Mcurve;
    \fill[green!20, opacity=0.6] (-2,1) rectangle (5,2);
  \end{scope}

  \begin{scope}
    \clip \Mcurve;
    \fill[red!20, opacity=0.6] (-2,2) rectangle (5,3);
  \end{scope}

  \begin{scope}
    \clip \Mcurve;
    \draw[dashed, thick] (-2,1) -- (5,1);
    \draw[dashed, thick] (-2,2) -- (5,2);
  \end{scope}

  \draw[thick] \Mcurve;
  \node at (1.5,3.3) {$\mathbb{H}(n)$};

  \node at (1.3,2.5) {$H_0$\footnotesize{ (VQC)}};
  \node at (0.8,1.5) {$H_1$};
  \node at (0.8,0.5) {$H_2$};

  \draw[->, thick] (4,1.5) -- (5,1.5) node[midway, above] {$\sim$};

  \node[draw, circle, fill=red!20, minimum size=3mm] (q3) at (6,2.6) {\footnotesize{$[\Lambda_0]$}};
  \node[draw, circle, fill=green!20, minimum size=3mm] (q2) at (6,1.5) {\footnotesize{$[\Lambda_1]$}};
  \node[draw, circle, fill=blue!20, minimum size=3mm] (q1) at (6,0.4) {\footnotesize{$[\Lambda_2]$}};
  \node at (6.2, 3.4) {$\mathbb{R}^{2^n} \cong \mathbb{H}(n) / \sim $};
\end{tikzpicture}
\caption{Observable space $\mathbb{H}(n)$ is partitioned into equivalent classes labeled by $\mathbb{R}^{2^n}$. In other words, $\mathbb{H}(n)$ is classified by its eigenvalues, and VQC falls in only one particular class.}
\label{fig: Hermitian equivalent classes}
\end{figure}
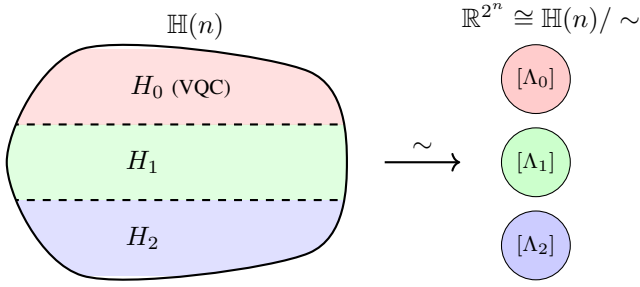

That is, for any $H \in \mathbb{H}(n)$ there exists $U \in \mcal{U}(\mcal{H}^n)$ and $\Lambda \in \mathbb{H}(n)$ such that
\begin{equation}\label{E: diag H}
    H = U^{\dagger} \Lambda U, \, \text{with} \,\,
\Lambda = 
\begin{pmatrix}
\lambda_1 & 0   & \cdots & 0 \\
0   & \lambda_2 & \cdots & 0 \\
\vdots & \vdots & \ddots & \vdots \\
0   & 0   & \cdots & \lambda_{2^n}
\end{pmatrix}, \,\, (\lambda_j \in \mathbb{R})
\end{equation}

By this observation, we propose using a diagonalization pair $(\Lambda, U)$ to represent an ANO of (\ref{E: non-local Hermitian}) and (\ref{E: H tensor product}), thereby reducing the number of Hermitian parameters from $K^2$ to $K$, which is called \textit{Diagonal Adaptive Non-local Observables (DANO)}. Formally,
\begin{equation}\label{E: DANO}
    \text{DANO} := \{ U^{\dagger} \Lambda U \, | \, \Lambda \text{:  diagonal} ,\, U \in \mcal{U}(\mcal{H}^k) \}
\end{equation}
where $U$ is realized by a conventional VQC with the extension of $K$ diagonal parameters for $\Lambda$. Consequently, DANO poses a simpler hardware challenge than ANO to relax the inherent fixed output bound of VQCs and approaches the same universality as ANO, namely,
\begin{equation}
  | \langle H_0 \rangle_{\text{VQC}} | \leq  | \langle H \rangle_{\text{DANO}} | = | \langle H \rangle_{\text{ANO}} |
\end{equation}
by the Rayleigh Quotient $\lambda_{\text{min}} \leq \langle \psi , Q \psi\rangle \leq \lambda_{\text{Max}} $ for $\| \psi \| = 1$ and $\lambda_{\text{min}} \leq \cdots \leq \lambda_{\text{Max}}$ eigenvalues of $Q$.

Obviously, DANO and ANO are mathematically identical. However, the equivalence on \textit{quantum hardware} requires additional conditions on the generating circuits.

Let $\text{DANO}(G) := \{ U^{\dagger} \Lambda U \, | \, \Lambda \text{:  diagonal} ,\, U \in \langle G \rangle \}$ where circuit $U$ comes from $\langle G \rangle$, the subgroup generated by a finite subset $G \subset SU(K)$.
\begin{theorem*}\label{Thm: DANO dense}
    $\mathrm{DANO}(G)$ is dense in $\mathrm{ANO}$ if $G\subset SU(K)$ is a finite set with $G^{-1}=G$, $\overline{\langle G\rangle} = SU(K)$ where $\overline{\langle G \rangle}$ is the topological closure of $\langle G \rangle$.
\end{theorem*}

\begin{proof}
For a non-zero $ H \in \text{ANO}$ with the decomposition $H = U^{\dagger} \Lambda U$, by the Generalized Solovay-Kitaev Theorem~\cite{SK-Thm} there exist constants $C_1 > 0$, $C_2 > 0$ and unitaries $U_1, \dots, U_L \in G$ such that,
\begin{equation}\label{E: Solovay-Kitaev}
    \bigl \|U - U_L U_{L-1} \cdots U_1 \bigr\| \leq \varepsilon / (2 \| \Lambda \| )
\end{equation}
with integer $L \leq C_1 \, (-\log {\varepsilon})^{C_2}$ and $\|\cdot\|$ the spectral norm $\| A \| = \sup_{ \| v \|=1 } \| A v \|$. On the other hand, we also have
\begin{equation}\label{E: Hermitian error}
\| H - V^{\dagger} \Lambda V  \| = \| U^{\dagger} \Lambda U - V^{\dagger} \Lambda V  \| \leq 2 \| \Lambda \| \| U - V \|
\end{equation}
whenever $\| U\|  = \| V \| = 1$. Taking $V = U_L U_{L-1} \cdots U_1 \in \langle G \rangle$, one has $\| H - V^{\dagger} \Lambda V  \| \leq \varepsilon$. Since $V^{\dagger}\Lambda V \in \mathrm{DANO}(G)$, it follows that $\mathrm{DANO}(G)$ is dense in $\mathrm{ANO}$.
\end{proof}

Therefore, DANO approximates ANO arbitrarily close, provided the generating set $G$ is good enough. A sharper bound of Solovay-Kitaev Theorem can be found in~\cite{kuperberg2023breaking} to refine the approximation guarantees above. 

In fact, as we demonstrate empirically below, augmenting a conventional VQC with a diagonal observable suffices to improve performance.

\section{Experiments}\label{sec_exp_results}

Two classification tasks are used to evaluate DANO with sliding $k$-local diagonal measurements: (1) Reduced \textbf{MNIST} for digit classification, and (2) face identification using a subset of the \textbf{Extended Yale Face Database B}.

\subsection{\textbf{Model Setup}}\label{subsec: DANO model setup}

DANO contains two variational components as in (\ref{E: DANO}): a diagonal observable $\Lambda(\lambda)$ for measurement and a variational unitary $U(\theta)$ implemented by a VQC. The trainable parameters are $\{ \lambda, \theta \}$.

\subsubsection{\textbf{VQC Architecture}}\label{subsubsec: VQC Architecture}
The VQC in the experiments consists of a date-encoding $V(x)$ and a variational $U(\theta)$ in (\ref{E: VQC}) (see Fig.~\ref{fig: VQC}). Denote the Pauli matrix by $\mathcal{P}=\{I,X,Y,Z\}$ and the corresponding \textit{rotations} $\{ R_x(\phi), R_y(\phi), R_z(\phi) \}$ of angle $\phi$. The encoding VQC is arranged as,
\begin{equation}\label{E: encoding V}
    V(x) = \bigotimes_{j=1}^n R_y(x_j) \circ \texttt{H}
\end{equation}
to take an input $x = (x_1, \ldots, x_n) \in \mathbb{R}^n$ and acting on an initial state $\ket{0}^{\otimes n}$ with $\texttt{H}$ a Hadamard gate. The variational gate is constructed as,
\begin{equation}\label{E: variational U}
    U(\theta) = \prod_{\ell=1}^L  \left( \bigotimes_{j=1}^n R
    _y(\theta_j^{(\ell)}) \circ  \mcal{C}^{(\ell)} \right)
\end{equation}
where $\theta = \{ \theta^{(\ell)}_j \} \in \mathbb{R}^{L\times n}$ are trainable parameters. Each $\mathcal{C}^{(\ell)}$ is a nearest-neighbor entangling layer on a 1D chain arranged in a brickwork pattern: CNOTs on $(1,2),(3,4), \ldots$ followed by a shifted layer on $(2,3),(4,5),\ldots$ (no periodic wrap-around). The layer structure is repeated $L = 6$ times in our experiments.

\subsubsection{\textbf{DANO Sliding Measurements}}\label{subsubsec: sliding k-local}

A family of $n$ trainable $k$-local \emph{diagonal} observables $\{\Lambda_1,\ldots,\Lambda_n\}$ is used for measurement. Each
$\Lambda_j=\mathrm{diag}(\lambda^{(j)}_1,\ldots,\lambda^{(j)}_{K})\in\mathbb{R}^{K\times K}$
has adaptable eigenvalues (as parameters) and acts on the $k$-qubit subset
$Q_j=(j,\ldots,j+k-1)\ (\mathrm{mod}\ n)$.
The corresponding expectation value is computed by inserting the identity $I$ on qubits \emph{not} in $Q_j$,
\begin{equation}\label{E: DANO output}
    z_j := \bra{0^{\otimes n}} V^{\dagger}(x) U^{\dagger}(\theta)\,
    \Big(I \otimes \cdots \otimes \underbracket{\Lambda_j}_{k} \otimes \cdots \otimes I \Big)\,
    U(\theta) V(x) \ket{0^{\otimes n}},
\end{equation}
so that $z:=(z_1,\ldots,z_n)\in\mathbb{R}^n$ is an $n$-dimensional output. Since $j$ ranges over all qubit indices $1,2,\ldots,n$, the observables are measured in a \emph{sliding} fashion that captures local structures. This scheme was introduced in \cite{lin2025adaptive}. For example, if $n=4$ and $k=3$, the cyclic $k$-qubit subsets are
\[
Q_1=(1,2,3),\quad Q_2=(2,3,4),\quad Q_3=(3,4,1),\quad Q_4=(4,1,2),
\]
see Fig.~\ref{fig: sliding k-local DANO}. The experiments below are both classifications tasks of 10 classes; thus $n=10$ is chosen to compare with various $k$-local DANO models.

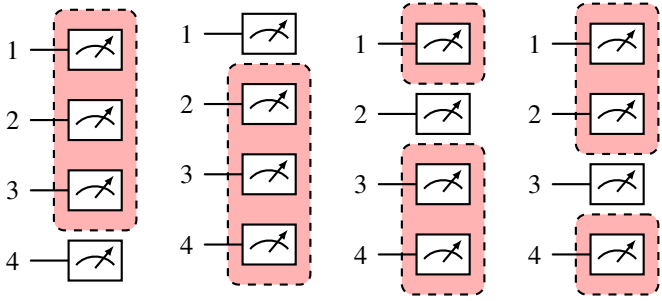
\begin{figure}[!h]
 \vskip -0.0in
  \centering
  \begin{tikzpicture}[scale=0.8, node distance=1.5cm]
    \node (circuit1) {
    \begin{quantikz}[row sep=0.4cm, column sep=0.5cm]
        \lstick{1}  &\meter{} \gategroup[3,steps=1,style={dashed,rounded corners,fill=red!30,inner xsep=2pt},background]{}\\
        \lstick{2} &\meter{} \\
        \lstick{3} &\meter{} \\
        \lstick{4} &\meter{}
    \end{quantikz}
    };
    \node (circuit2) [right=0.2cm of circuit1] {
    \begin{quantikz}[row sep=0.4cm, column sep=0.5cm]
        \lstick{1}  &\meter{} \\
        \lstick{2}  &\meter{} \gategroup[3,steps=1,style={dashed,rounded corners,fill=red!30,inner xsep=2pt},background]{}\\
        \lstick{3}  &\meter{} \\
        \lstick{4}  &\meter{}
    \end{quantikz}
    };
    \node (circuit3) [right=0.2cm of circuit2] {
    \begin{quantikz}[row sep=0.4cm, column sep=0.5cm]
        \lstick{1}  &\meter{} \gategroup[1,steps=1,style={dashed,rounded corners,fill=red!30,inner xsep=2pt},background]{}\\
        \lstick{2} &\meter{} \\
        \lstick{3}  &\meter{} \gategroup[2,steps=1,style={dashed,rounded corners,fill=red!30,inner xsep=2pt},background]{}\\
        \lstick{4}  &\meter{} 
    \end{quantikz}
    };
    \node (circuit4) [right=0.2cm of circuit3] {
    \begin{quantikz}[row sep=0.4cm, column sep=0.5cm]
        \lstick{1}  &\meter{} \gategroup[2,steps=1,style={dashed,rounded corners,fill=red!30,inner xsep=2pt},background]{}\\
        \lstick{2} &\meter{} \\
        \lstick{3}  &\meter{} \\
        \lstick{4}  &\meter{} \gategroup[1,steps=1,style={dashed,rounded corners,fill=red!30,inner xsep=2pt},background]{}
    \end{quantikz}
    };
  \end{tikzpicture}
  \caption{Sliding $k$-local DANO measurements of $k=3$, $n=4$.}
  \label{fig: sliding k-local DANO}
\end{figure}

\subsection{Exp 1: MNIST Digit Classification}

\subsubsection{\textbf{Dataset}}\label{subsubsec: MNIST Dataset}
The MNIST dataset contains grayscale images of handwritten digits (0–9), forming a 10-class classification task. Each image has an original resolution of $28 \times 28$. To reduce dimensionality for VQC processing, the images are downsampled to $4 \times 4$, corresponding to a 16-qubit input. A subset of 10,000 samples is used, split randomly into 9,000 training and 1,000 test images for evaluation.

\subsubsection{\textbf{Results}}\label{subsubsec: MNIST Results}

DANO is evaluated on MNIST with sliding $k$-local diagonal observables to compare with ANO. Across all settings, the VQC architecture is held fixed. The pure VQC with fixed Pauli measurements serves as the \textit{baseline}. Increasing $k$ raises measurement non-locality and improves model performance.

Fig.~\ref{fig: MNIST DANO Test Accu} shows the test accuracies over epochs, and Table~\ref{tab:mnist_dano} summarizes the best test accuracies and parameter counts. Since the VQC model is fixed across all settings, the observed improvements are attributable to the added variety from the dynamical diagonal observables $\Lambda(\lambda)$. 

The empty entries of ANO 6-local and 8-local in Table~\ref{tab:mnist_dano} are due to simulations exceeding the available 128GB RAM in our statevector implementation with Pennylane. This limitation is resulted from the classical memory/computational scaling of ANO, with details explained below.

\begin{figure}[htbp]
\vskip -0.1in
\begin{center}
\centerline{\includegraphics[width=1\columnwidth]{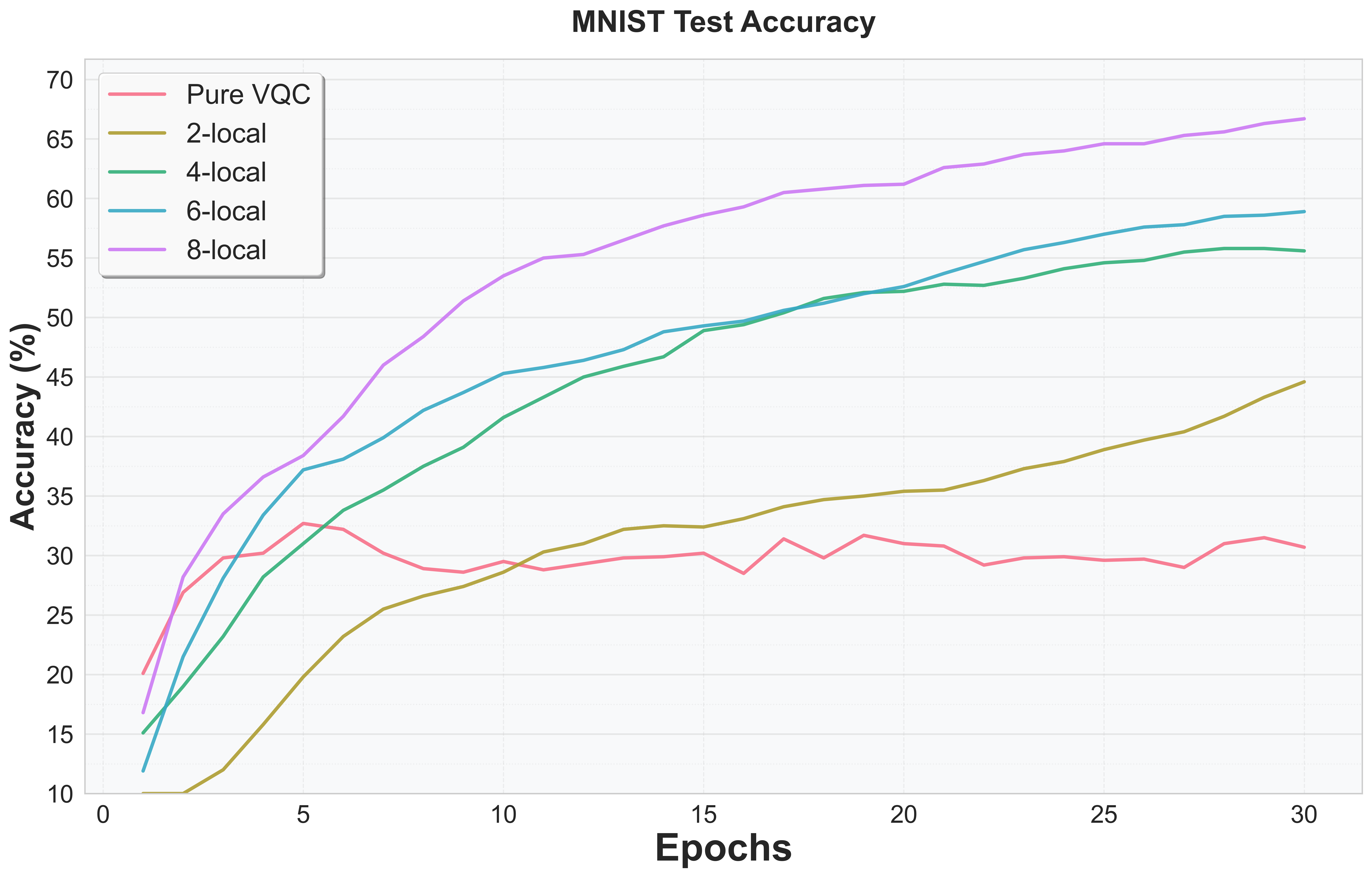}}
\caption{\textbf{[MNIST test accuracy.]} All models share the same VQC architecture. Increasing measurement non-locality improves test accuracy over the fixed-Pauli baseline (pure VQC).}
\label{fig: MNIST DANO Test Accu}
\end{center}
\vskip -0.1in
\end{figure}

\begin{table}[ht]
\centering
\caption{\textbf{$k$-local DANO on MNIST} from Fig.~\ref{fig: MNIST DANO Test Accu}. Empty entries of ANO 6-local and 8-local are due to the computing limit of classical simulation with 128Gb memory given.}
\label{tab:mnist_dano}
\vspace{2mm}
\begin{tabular}{p{0.14\linewidth} r c  r c}
    \toprule
    \textbf{Model} & \multicolumn{2}{c}{\textbf{DANO}} & \multicolumn{2}{c}{\textbf{ANO}}\\
    \cmidrule(r){2-5}
    {} & \textbf{Params $U + \Lambda$} & \textbf{Best accu} &\textbf{Params $U + H$}  & \textbf{Best accu} \\
    \cmidrule(r){1-5}
    {2-local}  & 160    & 44.6 \%           & 352        & 61.7 \%\\
    {4-local}  & 352    & 55.8 \%           & 4,192      & \textbf{75.2 \%}\\
    {6-local}  & 1,120   & 58.9 \%          & 65,632     & -\\
    {8-local}  & 4,192   & \textbf{66.7 \%} & 1,048,672    & - \\
    \cmidrule(r){1-5}
    {Pure VQC (baseline)} & 96     & 32.7 \%      & 96         & 32.7 \%\\
    \bottomrule
\end{tabular}
\end{table}

For DANO, the best test accuracy increases from $32.7\%$ (Pure VQC baseline) to $66.7\%$ (8-local),
\begin{itemize}
    \item VQC $\to$ 2-local: $+11.9\%$
    \item 2-local $\to$ 4-local: $+11.2\%$
    \item 4-local $\to$ 6-local: $+3.1\%$
    \item 6-local $\to$ 8-local: $+7.8\%$
\end{itemize}

Notably, the transition from DANO 4-local $\to$ 6-local yields only a modest gain ($+3.1\%$) despite a $4\times$ increase in diagonal measurement parameters ($256\to 1024$), suggesting partial saturation on downsampled MNIST.

Although Table~\ref{tab:mnist_dano} shows that 4-local ANO can exceed 8-local DANO in accuracy, DANO is intended to trade some accuracy at fixed small $k$ for substantially better \emph{computational scalability}. This advantage is especially relevant because the observable parameters are optimized on a classical computer rather than on quantum hardware. 

For $K=2^k$, a general $k$-local ANO observable carries $O(nK^2)=O(n4^k)$ trainable parameters, whereas DANO keeps only the diagonal eigenvalues and uses $O(n2^k)$. In the MNIST setting with $n=16$, ANO at $k=8$ requires $1,048,576$ real parameters over all sliding observables, while DANO requires only $4,096$, giving a $256\times$ reduction. 

In 16-qubit statevector simulation, the measurement-side forward/observable-gradient cost scales as $O(n2^{n+k})$ for ANO but only $O(n2^n)$ for DANO. At $k=8$, this corresponds to $268,435,456$ (ANO) flops vs. $1,048,576$ (DANO) over all observables. The optimizer state and memory overhead are reduced by the same factor. As $k$ increases, ANO therefore becomes rapidly impractical in classical simulation. As such, the 6-local and 8-local ANO exceeded the available 128 GB RAM in our implementation and could not be completed. This is precisely the regime where DANO remains computationally manageable. 

This trade-off is consistent with \textbf{Theorem }in Sec.~\ref{Sec: DANO}, where DANO approaches the full ANO when the circuit family generated by $G$ is sufficiently expressive to realize the desired eigenbases. The accuracy gap therefore suggests that the chosen low-depth gate set $G$ does not fully minimize that discrepancy between the ideal observable and the realizable Hermitian in (\ref{E: Hermitian error}).Therefore, some expressivity is sacrificed in exchange for large savings in classical memory and measurement-side computation. These exchanges make larger-$k$ measurements practical and significant in the next more challenging face-recognition experiment.

\subsection{Exp 2: Yale B- Face recognition}

\subsubsection{\textbf{Dataset}}\label{subsubsec: YaleB Dataset}
The Extended Yale Face Database B comprises grayscale face images acquired under controlled variations in illumination and pose. A partial dataset is constructed for a 10-class face identification task by randomly selecting 10 subjects (Fig.~\ref{fig: Yale B}) and retaining only non-ambient images under ``easy lighting'' conditions (defined by $|\text{azimuth}|<25^\circ$), yielding 1,980 images in total (198 per subject).

\begin{figure}[htbp]
\vskip -0.0in
\begin{center}
\centerline{\includegraphics[width=1\columnwidth]{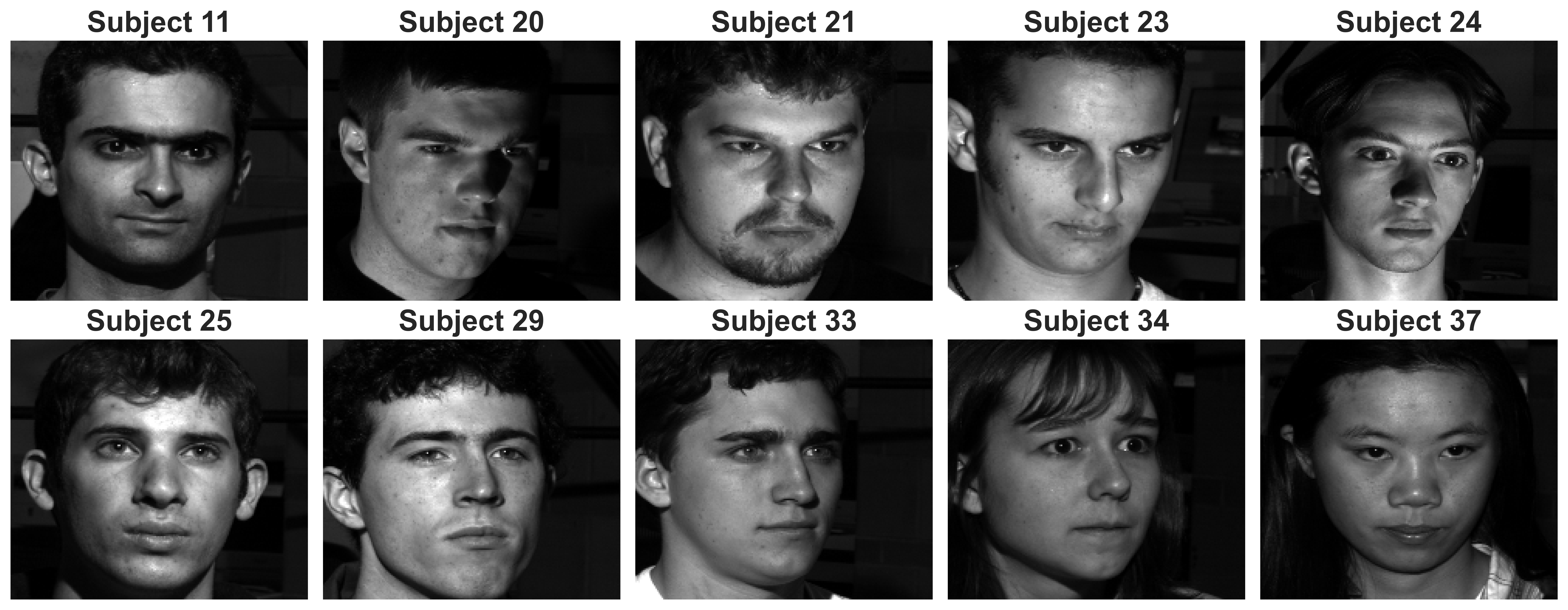}}
\caption{\textbf{Yale B}: 10 individuals selected for classification.}
\label{fig: Yale B}
\end{center}
\vskip -0.1in
\end{figure}

Each cropped image ($192\times168$ pixels) is flattened, standardized, and reduced to 16 dimensions via PCA. The resulting features are linearly rescaled to $[-\pi,\pi]$ for angle encoding, and a stratified split of 1,584/198/198 samples is used for training/validation/test, respectively.

The choice of PCA serves to reduce data dimension while introducing \textit{no additional} trainable parameters, unlike neural-network or CNN encoders whose learned weights can sculpt latent spaces to favor classification. To verify that the reduced 16-dimensional representation retains sufficient identity-related information, inverse-PCA reconstruction is performed for assessment, Fig.~\ref{fig: Yale B reconstruct}. The reconstructions remain visually reasonable and identity recognizable, indicating that the PCA features preserve salient structure and provide a meaningful latent representation for the subsequent VQC classification.

\begin{figure}[htbp]
\vskip -0.0in
\begin{center}
\centerline{\includegraphics[width=1\columnwidth]{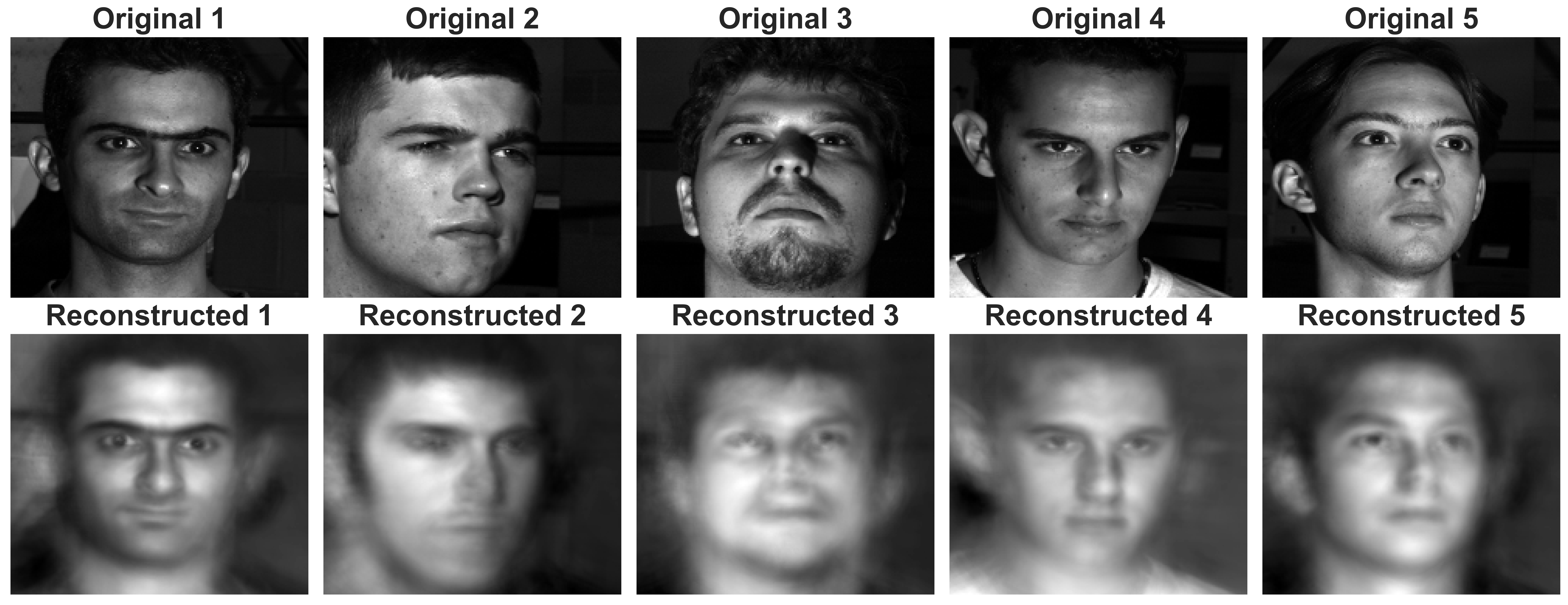}}
\caption{\textbf{Reconstruction} of first 5 individuals in Fig.~\ref{fig: Yale B}.}
\label{fig: Yale B reconstruct}
\end{center}
\vskip -0.15in
\end{figure}

\subsubsection{\textbf{Results}}\label{subsubsec: Yale B Results}

Applying the same DANO framework as in the MNIST experiment, non-local $k$ is varied to assess the effect with the VQC architecture remaining fixed.

Table~\ref{tab:yaleB_dano} shows a strong, monotonic improvement in test accuracy as the measurement locality $k$ increases, while the variational unitary remains fixed with $\#U(\theta)=96$. Accuracy rises from $30.3\%$ (Pure VQC) to $87.88\%$ (10-local), a $+57.58\%$ gain achieved solely through the diagonal Hermitian parameters. The corresponding increments are: 
\begin{itemize}
    \item Pure VQC $\to$ 4-local: $+21.72\%$
    \item 4-local $\to$ 6-local: $+12.63\%$
    \item 6-local $\to$ 8-local: $+12.12\%$
    \item 8-local $\to$ 10-local: $+11.11\%$.
\end{itemize}
These gains remain substantial even as $\#\Lambda(\lambda)$ grows rapidly from $0$ to $16{,}384$.

Compared with MNIST Table~\ref{tab:mnist_dano}, where improvements diminish beyond moderate $k$, Yale B benefits more consistently from increasing locality. This distinction perhaps reflects task complexity: face recognition under illumination variation exhibits richer, higher-order feature dependencies than downsampled digits, making broader $k$-local measurements more informative. As a result, DANO continues to extract useful correlations on Yale B long after gains saturate on MNIST.

\begin{figure}[htbp]
\vskip -0.0in
\begin{center}
\centerline{\includegraphics[width=1\columnwidth]{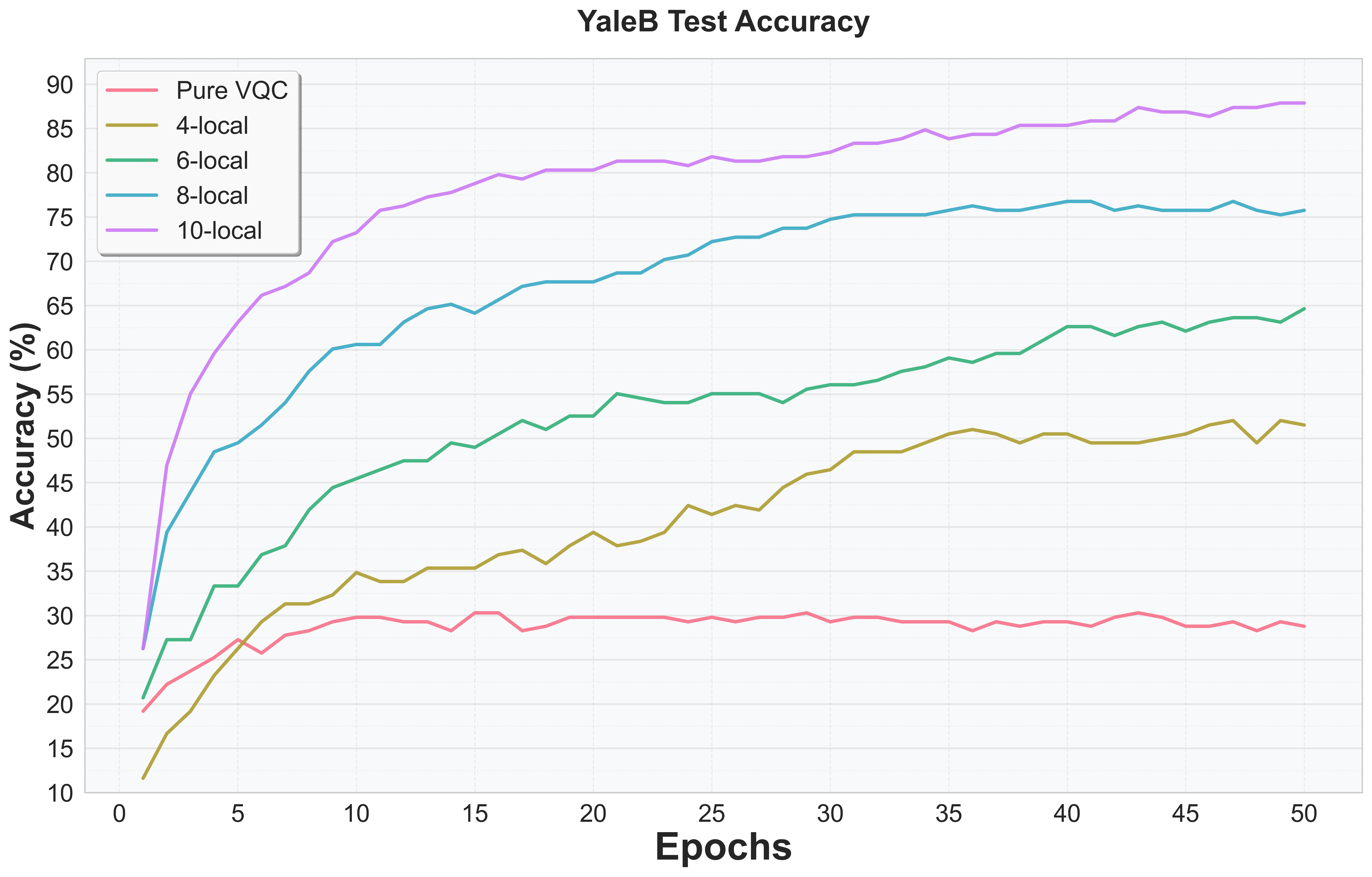}}
\caption{\textbf{Yale B test accuracy curves over epochs.} Performance steadily improved by increasing non-locality $k$.}
\label{fig: Yale B test accu}
\end{center}
\vskip -0.15in
\end{figure}

\begin{table}[ht]
\centering
\caption{\textbf{$k$-local DANO on Extended Yale B} from Fig.~\ref{fig: Yale B test accu}.}
\label{tab:yaleB_dano}
\vspace{-0mm}
\begin{tabular}{p{0.25\linewidth} c r c}
    \toprule
    \textbf{\# Model Params} & \textbf{$U(\theta)$} & \textbf{$\Lambda(\lambda)$} & \textbf{Best Test Acc.} \\
    \cmidrule(r){1-4}
    {Pure VQC}   & 96 & 0     & 30.3 \%\\
    {4-local}    & 96 & 256   & 52.0 \% \\
    {6-local}    & 96 & 1,024 & 64.7 \%\\
    {8-local}    & 96 & 4,096 & 76.8 \%\\
    {10-local}   & 96 & 16,384 & \textbf{87.9 \%} \\
    \bottomrule
\end{tabular}
\end{table}

\subsubsection{\textbf{Additional Results: Rescuing a VQC with DANO}}\label{subsubsec: Yale B Rescuing VQC}

To isolate the effect of adaptive measurements, a ``rescue'' protocol is conducted (Fig.~\ref{fig: Yale B Rescue VQC}). A pure VQC is trained for 50 epochs, yet its test accuracy remains low and quickly saturates (best test accuracy 30.30\% at epoch 15). The pure VQC model is then \textit{resuscitated} by switching to DANO at epoch 30: the unitary $U(\theta)$ is frozen at its epoch-30 weights, and only a diagonal observable $\Lambda(\lambda)$ is introduced and optimized thereafter, forming a DANO branch (pink line in Fig.~\ref{fig: Yale B Rescue VQC}). This intervention raises the test accuracy to 80.81\% by epoch 50, approaching the best 10-local DANO trained end-to-end from initialization (87.9\%) within a few percentage points. 

Consequently, the rescue outcome suggests that expanding the observable eigenvalues contribute larger on the performance than further optimizing the unitary evolution alone.

\begin{figure}[htbp]
\vskip -0.0in
\begin{center}
\centerline{\includegraphics[width=1\columnwidth]{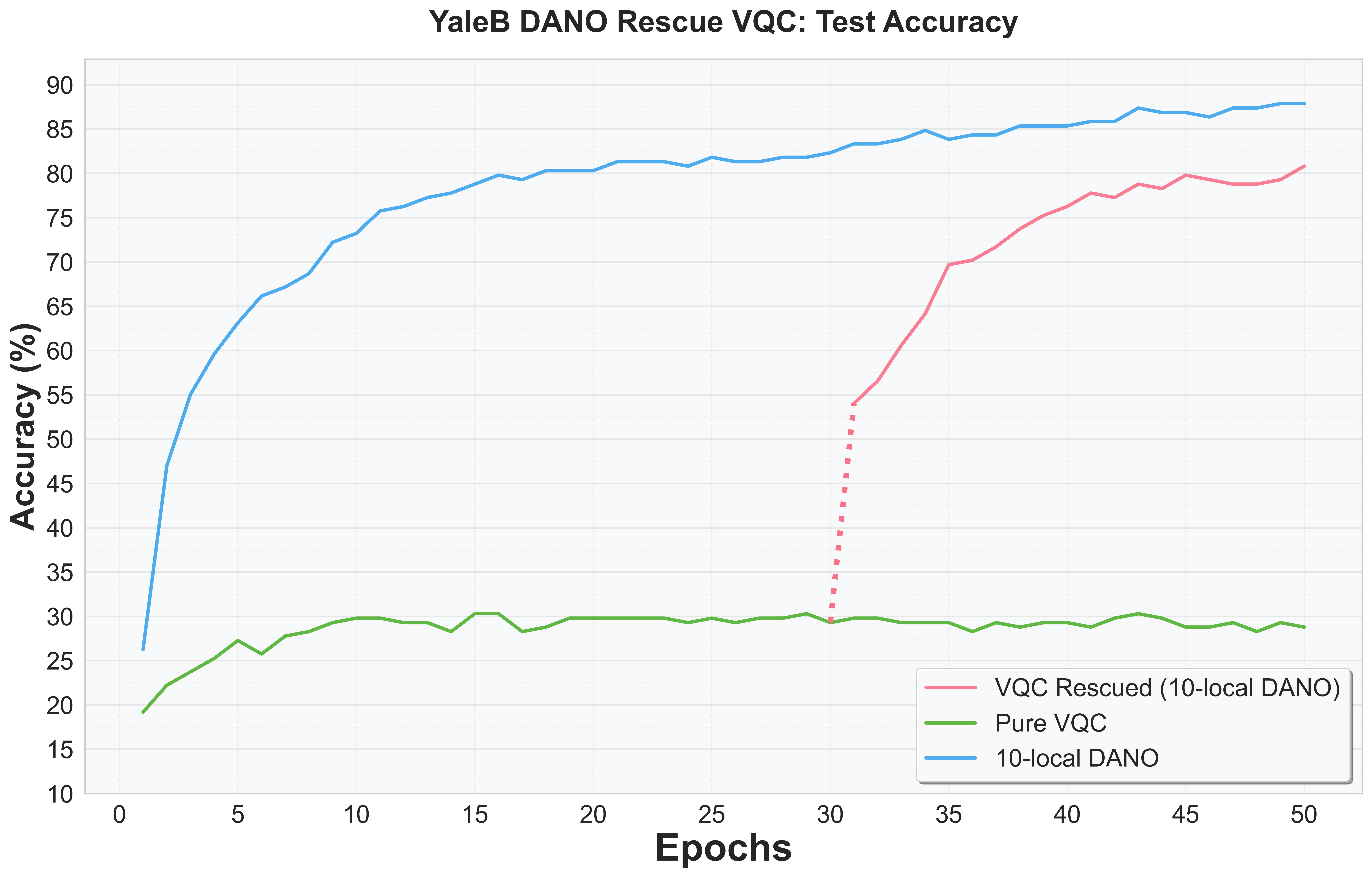}}
\caption{\textbf{Additional results}. The branch with dotted line is by switching from pure VQC to DANO of $k=10$ only the measurement eigenvalues to change. Eventually reaching 80.81\%, approaching the original 10-local 87.9\%.}
\label{fig: Yale B Rescue VQC}
\end{center}
\vskip -0.2in
\end{figure}

\section{Conclusion}\label{sec_conclusion}

We introduced DANO, in which a variational circuit $U(\theta)$ is paired with trainable diagonal observables $\Lambda(\lambda)$. 

Mathematically, diagonal observables serve as canonical representatives of ANO by modulo unitary similarity. Theorem in Sec.~\ref{Sec: DANO} guarantees DANO approaches ANO whenever the circuit family is sufficiently expressive (under the sense of Solovay-Kitaev). In particular, the conventional VQC also becomes a special case of DANO.

This diagonal representation by DANO yields a clear Classical-Quantum hybrid advantage, where the observable parameterization is reduced from $O(4^k)$ to $O(2^k)$ on $k$-local, and in statevector simulation the measurement step and observable-gradient computation drop from $O(2^{n+k})$ to $O(2^n)$. Thus, DANO trades certain measurement expressivity for substantially lower memory and training cost, making larger-$k$ regimes practical under fixed classical resources.

Empirically, sliding $k$-local DANO consistently improves over pure VQCs while keeping the unitary architecture fixed. Experimenting on reduced-sized MNIST, the best test accuracy improves from $32.7\%$ to $66.7\%$. On another Yale B face classification task, it rises from $30.3\%$ to $87.9\%$, with gains that remain strong at higher locality. A rescue experiment further shows that freezing a poorly performing VQC and optimizing only the diagonal observable can raise accuracy from $30.3\%$ to $80.8\%$, indicating that increasing observable spectrum contributes more than unitary optimization alone. 

Although ANO can outperform DANO at small $k$, the experiments and complexity analysis together show the intended trade-off, where DANO relinquishes certain accuracy in exchange for a substantial reduction in classical computational cost, while retaining major benefit of non-local adaptive measurement.

\bibliographystyle{IEEEtran}
\bibliography{references,bib/qml_examples,bib/vqc}

\end{document}